\documentclass[letterpaper, 10 pt, conference]{ieeeconf} 
\IEEEoverridecommandlockouts
\usepackage[utf8]{inputenc}
\usepackage{cite}

\usepackage{amsmath,amssymb,amsfonts,amsthm}
\usepackage{tabularx}
\usepackage{multirow}
\usepackage{booktabs}
\DeclareMathOperator*{\arginf}{arg\,inf}

\usepackage{algorithm}
\usepackage{algpseudocode}
\usepackage{graphicx}
\graphicspath{ {./figures/} }
\usepackage{textcomp}
\usepackage{xcolor}
\usepackage{multicol}
\newtheorem{theorem}{Theorem}[section]
\newtheorem{lemma}[theorem]{Lemma}
\newtheorem{remark}{Remark}
\newtheorem{definition}{Definition}

\newcommand\oprocendsymbol{\hbox{$\blacksquare$}}
\newcommand\oprocend{\relax\ifmmode\else\unskip\hfill\fi\oprocendsymbol}

\title{\LARGE \bf Protective Mission against a Highly Maneuverable Rogue Drone Using Defense Margin Strategy}
\author{Minjun Sung, Christophe J. Hiltebrandt-McIntosh, Hunmin Kim and Naira Hovakimyan
\thanks{This work has been supported by the National Science Foundation  (CNS-1932529), NASA  (NNH20ZEA001N-ULI, NNH21ZEA001N-USRC), and AFOSR.}%
\thanks{Minjun Sung, Christophe J. Hiltebrandt-McIntosh, and Naira Hovakimyan are with the Department of Mechanical Science and Engineering, University of Illinois at Urbana-Champaign, USA.
{\tt\small  \{mjsung2, cjh11, nhovakim\}@illinois.edu}}%
\thanks{Hunmin Kim is with the Department of Electrical and Computer Engineering, Mercer University, USA.
{\tt\small  kim$\_$h@mercer.edu }}
}

\begin{document}
\maketitle

\begin{abstract}
The current paper studies a protective mission to defend a domain called the \textit{safe zone} from a rogue drone invasion. We consider a one attacker and one defender drone scenario where only a noisy observation of the attacker at every time step is accessible to the defender. Directly applying strategies used in existing problems such as pursuit-evasion games are shown to be insufficient for this mission. We introduce a new concept of \textit{defense margin} to complement an existing strategy and construct a control strategy that successfully solves our problem. We provide analysis to point out the limitations of the existing strategy and how our defense margin strategy can enhance the performance. Simulation results show that our strategy outperforms that of the existing strategy at least by $36.0$ percentage points  in terms of protective mission success.
\end{abstract}

\section{Introduction}

Along with the rapid growth of the drone market, case reports and consequential concerns about malicious drones have been increasing. Some of the rogue drone incidents include flight interruption, terror attacks, privacy intrusion, and many more \cite{loeb2017exclusive,singh2018eye}. From these incidents, it becomes clear that we need to develop measures to defend property, land, or any valuable assets from unauthorized aerial invasions.  

Fortunately, there have been numerous types of research on tracking and defending moving aerial vehicles. In particular, various detection methods utilizing radio signals, radar, video, audio information, and combinations of these were investigated\cite{taha2019machine,guvenc2018detection}. Active capturing methods using net guns and birds have been proposed in \cite{guvenc2018detection}. Jamming methods using Electromagnetic Pulse (EMP) have been explained in\cite{shi2018anti}.    

Jamming is an effective neutralization scheme. However, since EMP can cause unintended impacts, most countries prohibit the use of jamming devices at the consumer level\cite{park2021survey}. As a result, we need to rely on aerial capturing methods where we need to shoot net bullets from the defender drones to neutralize the malicious drone\cite{park2021survey}. To do so, the defender drone must steer close enough to the rogue drone.

In constructing a problem to effectively steer the defender while protecting the \textit{safe zone}, we assume the followings: 1) The objective of the mission accounts for a designated area that we intend to protect. 2) The observation of a rogue drone is noisy. 3) The rogue drone is highly maneuverable and its trajectory or dynamics are not known in advance.  

A mission to track and defend a drone most resembles the close-in jamming problem introduced in\cite{valianti2021multi}. This work first addressed the problem to jam a rogue drone with observational uncertainty. However, it has limited the number of possible control actions for agents. More importantly, it differs from our problem in that it does not assume a protective mission and it chiefly focuses on jamming intensity. In our work, we instead assume a protective mission with an aerial capturing scenario where the defender drones have to approach the attacker drone more closely. Other relevant fields of research include differential Pursuit and Evasion (PE) games\cite{isaacs1999differential}, Perimeter Defense (PD) games\cite{shishika2020review}, and variations of these games. In PE games, a pursuer tries to intercept an evader while an evader tries to avoid a pursuer. This problem has been extensively studied including one pursuer-one evader problem\cite{ho1965differential,meier1969new,leitmann1968simple}, multi-agent problem\cite{hagedorn1976differential,fuchs2010cooperative, liu2013evasion, katz2005solution}, and under observational uncertainty\cite{shah2019multi,basimanebotlhe2014stochastic,pachter1983one,yavin1986pursuit}. This class of problems intends to find conditions for which the pursuer or evader can guarantee its victory, without considering any defense objective. A subclass of PE problem is Target-Attacker-Defender (TAD) game where an evader (attacker) additionally tries to reach the target while evading the pursuer (defender). TAD game assumes limited maneuverability of an attacker or knowledge of an attacker's dynamical model\cite{coon2017control,garcia2018optimal}. This is because the TAD game was initially motivated by military missions where such assumptions are reasonable. PD game, on the other hand, limits the defender to only move along the perimeter of the target during its mission, and it is without noise in the observation.

Our main contributions for this work are:
\begin{enumerate}
    \item Novel problem formulation assuming highly maneuverable agents, noisy observation, and the \textit{safe zone}. 
    \item Providing a new metric to quantify defense performance in a protective scenario. 
    \item Designing a defense strategy based on the new metric and proving its efficacy analytically and empirically.
\end{enumerate}

\subsection{Notations used in this work}
In this work we denote Euclidean norm as  $\|\cdot\|$, expectation of a random variable as $\mathbb{E[\cdot]}$. For geometric analysis in Section~\ref{sec:met}, we used $\overline{z_a z_b}$ to denote line segment that connects the two endpoints of the vectors $z_a$ and $z_b$. Moreover, $z_a \perp z_b$ and $z_a\parallel z_b$ respectively tell that the vectors are perpendicular and parallel to one another. 
{\renewcommand{\arraystretch}{1.2}
\begin{table}[ht]
\caption{Notations} \label{table:notations}
    \small
    \begin{tabularx}{\linewidth}{p{0.05\textwidth}X}
    \hline
    \hline
    $t$ & Time \\
    $\Omega_I$ & Zone of interest  \\
    $\Omega_S$ & Safe zone  \\ 
    $R_{\Omega_I}$ & Radius of $\Omega_I$\\
    $R_{\Omega_S}$ & Radius of $\Omega_S$\\
    $x^{a}_t$, $x^{d}_t$ & Attacker ($a$) and defender ($d$) state vector $\in \mathbb{R}^2$\\
    $e_t$ & Error vector $x^a_t - x^d_t$\\
    $u^{a}_t$, $u^{d}_t$ & Control input for attacker and defender\\
    $w_t$ & Measurement noise\\
    $y_t$ & Noisy observation of an attacker $x^a_t+w_t$\\
    $\sigma_t$ & Standard deviation of a measurement noise\\
    $\tau$ & Maximum capturing distance\\
    $\rho_{x^{a}_t}$ & Defense margin\\
    $\lambda_t$ & Weight parameters for defender control \\
    $P_t$ & Observational reliability\\
    \hline
    \hline
    \end{tabularx}
\end{table}}

\section{Problem formulation}

We consider a protective mission of a single defending drone, called the defender, against a single attacking drone, called the attacker. The objective of the defender is to prevent an attacker from invading the \textit{safe zone}. This paper aims to provide an effective defender strategy that can be implemented against a highly maneuverable attacker with an unknown trajectory and observational uncertainty. One defender and one attacker scenario can be considered as the smallest module which can be directly extended to multi-agent scenarios as in\cite{pierson2016intercepting,huang2011guaranteed}. 

\subsection{State-space representation} \label{sec: I}
The mission is assumed to be held in $\mathbb{R}^2$ space. The zone of interest $\Omega_I\subset \mathbb{R}^2$ is defined to be the region where observation of a drone in this area is considered to have a rogue intent. The safe zone $\Omega_S\subset \Omega_I$ is defined to be the domain that encompasses what the defender wishes to defend. The attacker wins the mission if it reaches $\Omega_S$ before getting intercepted by the defender. In this work, we assume $\Omega_I$ and $\Omega_S$ represent circles with respective radii $R_{\Omega_I}$ and $R_{\Omega_S}$, and the origin be the center of both circles.  

The attacker $a$ and the defender $d$ configuration at time $t$ are expressed as $x^{i}_t \in \mathbb{R}^2$ for $i \in \{a,d\}$. The configuration represents the planar position in Cartesian coordinates. Discrete-time dynamics of the attacker and the defender can be respectively written as:
\begin{equation}
    x^{i}_{t+1} = x^{i}_t + u^{i}_t, \quad i \in \{a,d\}.
\end{equation}
Here $u^{i}_t\in \mathbb{U}^{i}_t \subset \mathbb{R}^2$ for $i \in \{a,d\}$ is a deterministic control input of each agent, and $u^a_t$ is unknown to the defender at all times. Moreover, $\mathbb{U}^{a}_t$ and $\mathbb{U}^{d}_t$ denotes a set of admissible controls of the attacker and the defender at time $t$. 

In practice, $u^i_t$ for $i \in \{a,d\}$ can be considered as the \textit{speed} of agents, and they directly control the respective dynamics. In other words, we use a single integrator model in a discretized form. In this paper, we assume $\mathbb{U}^{d}_t=\mathbb{U}^{a}_t = \{u\in \mathbb{R}^2: \|u\| \leq 1\}$. Being able to instantly change the speed at any time, this condition accounts for the high maneuverability of drones. Moreover, $\|u^{i}_t\| \leq 1$ for $i \in \{a,d\}$ is for its normalization to respective maximum values. This is a relaxed assumption used in a handful of papers\cite{shah2019multi,garcia2019cooperative}, while others assume the defender to outpace the attacker\cite{makkapati2019optimal,von2020robust}. 

The attacker is considered to be \textit{intercepted} or \textit{captured} by the defender if the distance $\|e_t\| \triangleq \|x^a_t - x^d_t\|$ between the attacker and the defender is closer than the maximum capturing distance $\tau$. Formally, the attacker is intercepted if $\|e_t\| \leq \tau$. Choices of net guns characterizes the the maximum capturing distance $\tau$.

\subsection{Attacker detection model} 
In PE games with uncertainty, various models including Brownian motion model\cite{basimanebotlhe2014stochastic,yavin1986pursuit} and ellipsoid model\cite{shah2019multi} have been considered. In this work, we follow the uncertainty model used in\cite{davis2016c} such that we receive independent noisy state observation of the attacker at every time step.

An observation of $x^a_t$ at time $t$ is denoted as $y_t \in \mathbb{R}^2$, and is subject to a zero-mean Gaussian noise with covariance matrix $\sigma^2_t I_2$, where $\sigma ^2_t \in \mathbb{R}_{\geq0}$ represents a variance of a Gaussian distribution, and $I_2\in \mathbb{R}^{2\times 2}$ represents an identity matrix \cite{morbidi2012active},\cite{valianti2021multi}. Formally, the following model is adopted to express the observational uncertainty:
\begin{equation}\label{eq: distribution}
\begin{aligned}
    y_t &= x^a_t + w_t \\
    w_t &\sim \mathcal{N} (0,\sigma^2_t I_2).
\end{aligned}
\end{equation}

Lastly, $\sigma_t$ is modeled by adopting the uncertainty model proposed in\cite{davis2016c}:
\begin{equation}\label{eq: sigma}
    \sigma^2_t = \beta_{b}+ \beta_{d} \|e_t\|^2+ \beta_{v}  (1-\nu_t).\\
\end{equation}
Parameters $\beta_b, \beta_d,\beta_v$ are non-negative real values characterized by the sensor and the estimation model. Specifically, they represent the variance coefficient for baseline, distance, and visibility, respectively. Visibility $\nu_t\in [0,1]$ relates the blockage of the sight to the variance of the uncertainty, such that $\nu_t =0$ if the sight is fully blocked by an obstacle, and $\nu_t=1$ if the sight is not blocked at all. Any values between represent partial blockage of the sight.

In this paper, we will consider an environment without any obstacles such that $\nu_t \equiv 1$. Furthermore, we will consider zero baseline variance or $\beta_b = 0$ implying that the observational uncertainty becomes zero when the distance $\|e_t\|$ is zero. Then, we can rewrite \eqref{eq: distribution} as
\begin{equation}\label{uncertainty}
\begin{aligned}
    y_t &= x^{a}_t +w_t\\ 
    w_t \sim \mathcal{N}& (0,\beta \|e_t\|^2 I_2),
\end{aligned}
\end{equation}
where $\beta$ is a short hand notation for $\beta_d$.

\subsection{Joint tracking and defending problem}

Now we formally state our problem in this section. The defender's mission is to prevent the attacker from landing at the safe zone $\Omega_S$ for all time \textit{or} to intercept the attacker before it reaches the safe zone. Precisely, the problem is to find discrete control input $u^{d}_t$ such that it satisfies
\begin{equation} \label{condition1}
\begin{aligned}
    x^{a}_t \notin \Omega_S &\: \forall t \in [t_i,t_f]\\
    &\text{ Or }\\
    \exists t_c\in [t_i,t_f]:  (x^{a}_t \notin \Omega_S \: &\forall t \in [t_i,t_c])\wedge  (\|e_{t_c}\| \leq \tau)
\end{aligned}
\end{equation}
subject to
\begin{equation}\label{condition1-subject}
\begin{aligned}
    \|u^{d}_t\| \leq 1, \|&u^{a}_t\| \leq 1 \quad \forall t \in [t_i,t_f]\\
    y_t &= x^{a}_t +w_t\\
    w_t \sim \mathcal{N}& (0,\beta \|e_t\|^2 I_2)
\end{aligned}
\end{equation}
where $t_i$, $t_f$, $t_c$ respectively denote the initial time of observation, terminal time that can be chosen by the user, and the capturing time. Note that this problem is not limited to the interception problem, but defines a more general class of a defense problem. The defender can win also by not letting the attacker pass through for a sufficiently long \textit{runtime}. Fig~\ref{fig:Figure1} visualizes the problem. 

\begin{figure}[ht]
    \centering
    \includegraphics[width=\linewidth]{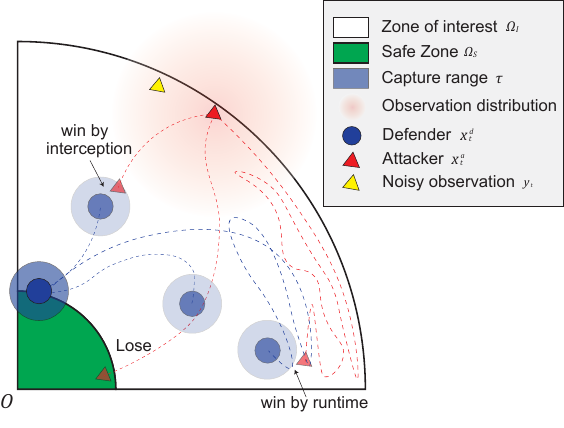}
    \caption{Problem visualization}
    \label{fig:Figure1}
\end{figure}

\section{Method}

Our solution to the joint tracking and defending problem is partly motivated by the properties of the Pure Pursuit  (PP) strategy, which is a widely adopted guidance law for interception missions. We will formally introduce and explain the advantages and limitations of the PP strategy in Section~\ref{sec:PP} along with other popular guidance laws. In Section~\ref{sec:met} we introduce a strategy based on \textit{defense margin} which can complement the PP strategy. Then, in Section~\ref{sec:TV}, we propose a strategy that combines the two strategies to effectively solve our problem.

\subsection{Baseline: Pure Pursuit strategy}\label{sec:PP}
Typical and popular strategies utilized in PE games are Constant Bearing (CB), Line of Sight (LoS), and PP guidance laws~\cite{breivik2008guidance}. CB assumes the knowledge of the attacker's instantaneous velocity as well as its position~\cite{makkapati2018pursuit}, whereas the defender only has access to the noisy observation of the attacker in our problem. Consequently, CB is not suitable for application to our problem. LoS, on the other hand, is known to be infeasible in missions with observational uncertainty unless there are external or additional measures to complement the noisy observation\cite{ratnoo2011line}. 

Having only access to the instantaneous positional estimate of the attacker, the PP strategy is a reasonable strategy to be considered\cite{makkapati2018pursuit}. The idea of this strategy is to always steer the defender directly to the observation of the attacker \eqref{condition1}. Formally, the defender's control input is designed by
\begin{equation} \label{pp_control}
    u^{d}_t = \frac{y_t-x^{d}_t}{\|y_t-x^{d}_t\|},
\end{equation} 
where $u^{d}_t$ is normalized to meet the constraint \eqref{condition1-subject}.

In this work, we show that the PP strategy is effective, but for limited conditions due to the presence of uncertainty. Here we explain such conditions analytically.

\begin{definition}
Consider the \textit{n}-dimensional stochastic discrete time system
\begin{equation}\label{eq: dsct}
    \zeta_{t+1} = f (\zeta_t,\chi_t,\chi), \quad \zeta (t_0) = \zeta_0 
\end{equation} The trivial solution of the system is said to be stochastically stable or stable in probability, if  for every $\epsilon > 0$ and $h>0$ there exists $\delta = \delta (\epsilon,h,t_0)>0 $ such that 
    \begin{equation}
        P\{|\zeta_t|<h\}\geq 1-\epsilon, \quad t\geq t_0
    \end{equation}
when $|\zeta_0|<\delta$. Otherwise, it is said to be stochastically unstable\cite{li2013stability}.
\end{definition}

Consider the Lyapunov function $V: \mathbb{R}^n \rightarrow \mathbb{R}$, with $V (0)=0$. Its discrete increment it is expressed as follows: 
\begin{equation} \label{deltalyapunov}
    \Delta V (\zeta_t) = V (\zeta_{t+1})-V (\zeta_t)
\end{equation}
Using this definition and notation of discrete Lyapunov function and its increment, the following theorems are derived:
\begin{theorem}
If there exists a positive definite function $V (\zeta_t)\in C^2 (D_r)$, such that 
\begin{equation}\label{lyapunouv_theorem}
E[\Delta V (\zeta_t)] \leq 0    
\end{equation}
for all $\zeta_t \in D_r$, then the trivial solution of \eqref{eq: dsct} is stochastically stable in probability\cite{li2013stability}.
\end{theorem}

If we consider $\zeta_t$ in \eqref{lyapunouv_theorem} to be $e_t$, we can interpret the stochastic stability of $e_t$ as the expected defender state converging to that of the attacker, meaning interception. In the following, we provide the condition that guarantees such convergence when using the PP strategy.

\begin{theorem}\label{theorem:main1}
Assume $\|e_t\|>\sqrt{2}$ and $\|w_t\| < \|e_t\|$. The error $e_t$ is stable in probability under the PP strategy \eqref{pp_control} if the following condition holds: 
\begin{equation}\label{maintheorem1}
    \frac{e_t^\top u^{a}_t+1}{\|e_t+u^a_t\|} \leq \mathbb{E}[\cos \alpha],
\end{equation}
where $$\cos \alpha \triangleq \frac{ (e_t+w_t)^\top (e_t+u^{a}_t)}{\|e_t+w_t\|\: \|e_t+u^{a}_t\|}, \: \alpha \in  (-\frac{\pi}{2}, \frac{\pi}{2}).$$

\end{theorem}
\begin{proof}

\begin{figure}[h]
    \centering
    \includegraphics[width=0.6\linewidth]{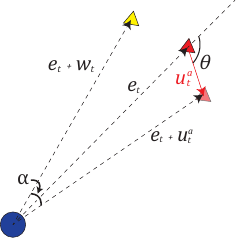}
    \caption{Visualization of notations and notions used in Theorem~\ref{theorem:main1}}
    \label{fig:Figure2}
\end{figure}

We use the Lyapunov function to provide a condition under which stability can be guaranteed.

Define a Lyapunov function $V (e_t)$ as a dot product of $e_t$ to itself:
\begin{align}\label{eq:Lyap}
    V (e_t) = e^\top _t e_t.
\end{align}

By construction, $V (e_t)$ is positive definite, and $V (0)=0$. Plugging the uncertain observation model \eqref{uncertainty} into the PP control \eqref{pp_control} yields:

\begin{equation}\label{eq:input}
\begin{aligned}
    u^{d}_t &= \frac{y_t-x^{d}_t}{\|y_t-x^{d}_t\|} = \frac{e_t + w_t}{\|e_t + w_t\|}
\end{aligned}
\end{equation}

Plugging~\eqref{eq:input} into~\eqref{deltalyapunov}, we have
\begin{equation}\label{lyapunov_expanded}
\begin{aligned}
    \Delta V (e_t) &= V (e_{t+1})-V (e_t)\\
    &=  (e_t - \frac{e_t+w_t}{\|e_t+w_t\|} +u^{a}_t)^\top  (e_t - \frac{e_t+w_t}{\|e_t+w_t\|} +u^{a}_t)\\
    &\quad - e^\top _t e_t\\
    &= -2\frac{e^\top _t (e_t+w_t)}{\|e_t+w_t\|} + 2e^\top _t u^{a}_t + \frac{ (e_t+w_t)^\top  (e_t+w_t)}{\|e_t+w_t\|^2}\\
    &\quad +u^{a\top} _t u^{a}_t - 2\frac{ (e_t+w_t)^\top u^{a}_t}{\|e_t+w_t\|}
\end{aligned}
\end{equation}

Rearranging and taking expectation on both side yields,

\begin{equation}\label{generalinequality}
    \begin{aligned}
    \mathbb{E}[&\Delta V (e_t)] = \\
    &\quad -2\mathbb{E}[\frac{ (e_t+w_t)^\top (e_t+u^{a}_t)}{\|e_t+w_t\|}] + \mathbb{E}[\frac{ (e_t+ w_t)^\top  (e_t +w_t)}{\|e_t+w_t\|^2}]\\
    &\quad +2\mathbb{E}[e^\top _t u^{a}_t] + \mathbb{E}[u^{a\top} _t u^{a}_t]\\
    &\leq -2\mathbb{E}[\frac{\|e_t+w_t\|\;\|e_t+u^{a}_t\|\; \cos \alpha}{\|e_t+w_t\|}]\\
    &\quad + 2e_t^\top u^a_t +2\\
    &= -2\|e_t+u^{a}_t\|\mathbb{E}[\cos \alpha] + 2e_t^\top u^a_t +2
    \end{aligned}
\end{equation}
Here we simply used $\|u^a_t\| \leq 1$. In addition, $\alpha\in  (-\frac{\pi}{2},\frac{\pi}{2})$ due to $\|e_t\|>\sqrt{2}$ and $\|w_t\|<\|e_t\|$. 

Rearranging \eqref{generalinequality} to satisfy \eqref{lyapunouv_theorem}, we obtain \eqref{maintheorem1}, completing the proof.

\end{proof}

\begin{remark}
Stochastic stability of trivial case  ($u^a_t$ being a zero vector) can be directly proved after \eqref{lyapunov_expanded} simply by plugging in $u^a_t$ a zero vector and using $\|e_t\|>\sqrt{2}$ and $\|w_t\|<\|e_t\|$.
\end{remark}

For the PP strategy to be effective, we need \eqref{maintheorem1} to hold. To illustrate this point, implicitly define $\theta$ by 
$$\cos \theta \triangleq \frac{e_t^\top u^{a}_t}{\|e_t\|\cdot\|u^{a}_t\|}, \: \theta \in [-\pi, \pi].$$
Consider $\cos \theta = -1$ and $\|u^{a}_t\| = 1$ such that $e_t^\top u^a_t = -\|e_t\|$, for which the attacker is moving directly towards the defender. Then, $\frac{-\|e_t\|+1}{\|e_t\|-1} =  -1$, which makes \eqref{maintheorem1} to be always true regardless of $\alpha$. On the other hand, one can also find out that large $\cos \alpha$ can be obtained when we have sufficiently small $\|e_t\|$ in addition to $\cos \theta \simeq -1$. That is because small $\|e_t\|$ will yield $e_t+w_t \simeq e_t$ by \eqref{condition1-subject}, and $\cos \theta \simeq -1$ will yield $\kappa e_t \simeq e_t+u^a_t$ where $\kappa\in  (0,1]$ is a constant. This consequently makes $\cos \alpha \simeq 1$ to make the inequality to hold. 

The PP strategy becomes sufficiently effective for interception when $\cos \theta \simeq -1$. Since the defender does not know the precise position of the attacker, small $\|e_t\|$ to induce $w_t \simeq 0$ needs to be satisfied in advance. In other words, the defender would have to behave in a \textit{conservative} manner until small $\|e_t\|$ is achieved and then utilize the PP strategy.

\begin{remark}\label{remark:theta0}
Note that if $\theta = 0$, the attacker is heading directly away from the defender. For $\|u^a_t\| = 1$ the left-hand side of \eqref{maintheorem1} becomes $1$. The inequality does not hold \textit{almost surely}. This agrees with our intuition that if the attacker is moving away from the defender, the best pursuit a defender can do is to keep $\|e_t\|$ constant, as long as the maximum speed of a defender and an attacker are equivalent. 
\end{remark}

\subsection{Complement: Defense Margin Strategy}\label{sec:met}

The limited reliability of the PP guidance law makes it insufficient to be applied to our mission. In particular, the goal of the PP strategy corresponds only to the second objective in \eqref{condition1}. We intend to design a strategy that accounts for both. In this subsection, we explain a safe reachable set and apply this to suggest a new metric \textit{defense margin} that measures a defense performance at each state. Then, we introduce a Defense Margin strategy (DM strategy) and explain how this can complement the PP strategy.

\begin{definition}
The safe reachable set $L_{x^a_t}$ is the set of positions reachable by the attacker before the defender\cite{shah2019multi}. 
\end{definition}
Following the assumption in \eqref{condition1-subject} that the defender is at least as fast as the attacker, we can express $L_{x^a_t}$ as follows:
\begin{equation}
    L_{x^a_t} = \{l\in \mathbb{R}^2|\ \|l-x^a_t\|\leq \|l-x^{d}_t\| \}. \label{srs}
\end{equation}

Geometrically, the safe reachable set is the half-plane, points which are closer to the attacker than the defender.

Having \eqref{srs}, we can subsequently define $l_{x^a_t} \in L_{x^a_t}$ as the closest point in the reachable set to the safe zone:
\begin{equation}\label{lxat}
    l_{x^a_t} = \arginf_{l \in L_{x^a_t}} \|\Omega_S-l\|,
\end{equation}
where $\|\Omega_S-l\| \triangleq \inf_{\omega_S \in \Omega_S}\|\omega_S - l\|$, for a given $l$.

Finally, we can define a new metric \textit{defense margin}.
\begin{definition}
Defense margin $\rho_{x^a_t}$ is the norm of $l_{x^a_t}$: 
\begin{equation}
    \rho_{x^a_t} = \|l_{x^a_t}\|.
\end{equation}
\end{definition}

Note that if $\rho_{x^a_t}\leq R_{\Omega_S}$, there exists a strategy for the attacker to reach the $\Omega_S$ regardless of the defender's strategy. 

Let
\begin{equation} \label{new_control}
    u^{d}_t = \frac{l_{y}-x^{d}_t}{\|l_{y}-x^{d}_t\|},
\end{equation}
where $l_y$ is defined by replacing $x^a_t$ with $y_t$ in \eqref{lxat}. We will refer to $ u^{d}_t$ as {\em Defense Margin} strategy.

Intuitively, the DM strategy makes the defender maneuver to the closest point from the safe zone that the attacker can potentially reach. This can be considered as a strategy to implicitly accomplish the first goal of \eqref{condition1} by enforcing the attacker to take a detour to reach the safe zone. 

\begin{lemma}\label{lemma:DM}
The Defense margin $\rho_{x^{a}_t}$ can be measured with the following equation
\begin{equation}
    \rho_{x^{a}_t} = \frac{1}{2}\frac{\|x^{a}_t\|^2-\|x^{d}_t\|^2}{\|x^{a}_t-x^{d}_t\|}.
\end{equation}
\end{lemma}
\begin{proof}

Recall that $l_{x^{a}_t}$ is a vector in the half-plane $L_{x^{a}_t}$ that has a minimum distance to the origin. Moreover, $ (\frac{x^{d}_t+x^{a}_t}{2}-l_{x^{a}_t}) \perp  (x^{a}_t-x^{d}_t)$ or equivalently, 
\begin{equation}
     (\frac{x^{d}_t+x^{a}_t}{2}-l_{x^{a}_t})\cdot  (x^{a}_t-x^{d}_t) = 0
\end{equation}
which yields
\begin{equation}
     (\frac{x^{d}_t+x^{a}_t}{2})\cdot  (x^{a}_t-x^{d}_t) = l_{x^{a}_t}\cdot  (x^{a}_t-x^{d}_t).
\end{equation}
Equivalently,
\begin{equation}
    \frac{\|x^{a}_t\|^2-\|x^{d}_t\|^2}{2} = \|l_{x^{a}_t}\| \|x^{a}_t-x^{d}_t\|
\end{equation}
where the right hand side holds since $l_{x^{a}_t} \parallel x^{a}_t-x^{d}_t$.
Solving for $\|l_{x^a_t}\|$ yields,
\begin{equation}
    \rho_{x^{a}_t} = \|l_{x^a_t}\| = \frac{1}{2}\frac{\|x^{a}_t\|^2-\|x^{d}_t\|^2}{\|x^{a}_t-x^{d}_t\|}.
\end{equation}
\end{proof}

In the following, we provide analytical proof to explain that \eqref{new_control} outperforms \eqref{pp_control} in terms of defense margin, implying that the DM strategy can complement the PP strategy.

\begin{theorem}\label{theorem:main2}
Assume $\|e (t)\|> \sqrt{2}$ and $\|x^a_t\| > \|x^d_t\|$. For static attacker state vector $x^a_{t+1} = x^a_t$ with uncertainty $w_t = 0 \: \forall t$, the following inequality holds for one step change of the defense margin:
\begin{equation}\label{maintheorem2}
\Delta \rho_{x^a_t}|u^d_{DM} \geq \Delta \rho_{x^a_t}|u^d_{PP},
\end{equation}
where $\Delta \rho_{x^a_t}|u^d_{DM}$ and $\Delta \rho_{x^a_t}|u^d_{PP}$ respectively denote $\Delta \rho_{x^a_t}$ following the DM strategy \eqref{new_control} and the PP strategy \eqref{pp_control}.
\end{theorem}

\begin{proof}
The proof is explained in three blocks: 
\begin{enumerate}
    \item Transform coordinates to simplify the configuration.
    \item Show that the change in defense margin for the PP strategy is precisely $\frac{1}{2}$, or formally $\Delta \rho_{x^a_t}|u^d_{PP} \equiv \frac{1}{2}$.
    \item Show that $\Delta \rho_{x^a_t}|u^d_{PP} \geq \frac{1}{2} $.
\end{enumerate}
\vspace{5pt}

\begin{figure}[h]
    \centering
    \includegraphics[width=0.8\linewidth]{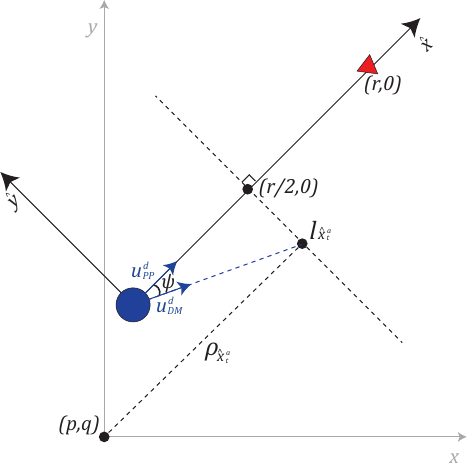}
    \caption{Visualization of notations and notions used in Theorem~\ref{theorem:main2}}
    \label{fig:Figure3}
\end{figure}

\textbf{1) Coordinate Transformation}
At time $t$ given $x^a_t$ and $x^d_t$, we do a \textit{rigid} coordinate transformation $\Phi: x \rightarrow \hat{x}$ such that $x^d_t \rightarrow \hat{x}^d_t = [0,0]^\top$, $x^a_t \rightarrow \hat{x}^a_t = [r,0]^\top$ where $r = \|e (t)\|$ and the center of the safe zone $\Omega_S$ will correspondingly be transformed to $ [p,q]^\top$. The assumption $\|x^a_t\| > \|x^d_t\|$ is translated to
\begin{equation}
    p < \frac{r}{2}
\end{equation}
in the transformed coordinate.

Rigid transformation only allows rotation and is followed by translation, and therefore preserves the Euclidean distance between every pair of points. In this transformed coordinate, the two strategies are simplified as $u^d_{PP} = [1,0]^\top$, and $u^d_{DM} = [\cos \psi,\sin\psi]^\top$. In other words, $x^d_{t+1} = [1,0]^\top$ for PP strategy, and $x^d_{t+1} = [\cos\psi, \sin \psi]^\top$ for DM strategy. 

Now $\rho_{\hat{x}^a_t}$ lies precisely on the perpendicular bisector of the line segment $\overline{\hat{x}^a_t \hat{x}^d_t}$. Consequently, $\angle \rho_{\hat{x}^a_t}\hat{x}^d_t \hat{x}^a_t = \angle \rho_{\hat{x}^a_t}\hat{x}^a_t \hat{x}^d_t = \psi$ where $\psi \in  (-\frac{\pi}{2}, \frac{\pi}{2})$. Formally, following holds:
\begin{equation}\label{p}
     p = \frac{r}{2}-\rho_{\hat{x}^a_t} < \frac{r}{2}
\end{equation}
\begin{equation}\label{q_psi}
    q = \frac{r}{2}\tan \psi.
\end{equation}

\textbf{2) Change in Defense Margin for PP strategy }

From \eqref{p} we have $\rho_{\hat{x}^a_t} = \frac{r}{2} - p$. For $x^d_{t+1} = [1,0]^\top$, we obtain $\rho_{\hat{x}^a_{t+1}} = \frac{r+1}{2} - p$. Consequently,
\begin{equation} \label{lambda0}
\begin{aligned}
    \Delta \rho_{\hat{x}^a_t}|u^d_{PP} &= \rho_{\hat{x}^a_{t+1}}-\rho_{\hat{x}^a_t} \\
    &=  (\frac{r+1}{2}-p) -  (\frac{r}{2}-p) = \frac{1}{2}.
\end{aligned}
\end{equation}

\textbf{3) Change in Defense Margin - Our strategy }

Similar to \eqref{lambda0}, to obtain $\rho_{\hat{x}^a_{t+1}}$ we will need to find a distance from $[p,q]^\top$ to the straight line that bisects $\hat{x}^d_{t+1} = [\cos \psi, \sin \psi]^\top$ and $[r,0]^\top$ or equivalently, 
\begin{equation}\label{bisecting_line}
    \frac{r-\cos \psi}{\sin \psi} (x-\frac{r+\cos \psi}{2}) -  (y-\frac{\sin \psi}{2}) = 0.
\end{equation}
The distance from $[p,q]^\top$ to \eqref{bisecting_line} can be obtained with the following equation:
\begin{equation}\label{rholambda1}
    \rho_{\hat{x}^a_{t+1}}|u^d_{DM} = \frac{|r^2-1+2q\sin \psi -2p (r-\cos\psi)|}{2\sqrt{ (r-\cos\psi)^2 + \sin ^2 \psi}}.
\end{equation}

Here we will find the lower bound of $\rho_{\hat{x}^a_{t+1}}|u^d_{DM}$ and assert that it is greater or equal to $\frac{1}{2}$ to complete our proof.

First observe that we only need to consider for $\psi \in [0,\frac{\pi}{2})$ since $\rho_{\hat{x}^a_{t+1}}$ is symmetrical with respect to $\hat{x}$-axis. Every result we obtain can therefore be identically proved for $\psi \in  (-\frac{\pi}{2}, 0]$. Subsequently, $q\geq 0$ by \eqref{q_psi}. Note that $q=0$ or $\psi = 0$ is a trivial case which yields $u^d_{PP} \equiv u^d_{DM}$. This obviously makes $\Delta\rho_{x^a_t}|u^d_{PP} \equiv \Delta \rho_{x^a_t}|u^d_{DM} = \frac{1}{2}$. 

The numerator of \eqref{rholambda1} can be lower-bounded by the following:
\begin{equation}
\begin{aligned}
    r^2-1 + &2q\sin \psi -2p (r-\cos \psi) \\
    &> r^2-1 + 2q\sin \psi -2p (r-\cos \psi) |{p=\frac{r}{2}}\\
    &= -1 + 2q\sin \psi +r\cos \psi \\
    &= -1 + \frac{r}{\cos \psi} > 0 \quad \forall \psi \in [0,\frac{\pi}{2}).
\end{aligned}
\end{equation}
Here we used $\psi \in [0, \frac{\pi}{2})$, $r>\sqrt{2}$, $q = \tan \psi$, and $p< \frac{r}{2}$. Hence, we can ignore the absolute operator for the remainder of the proof.

Rewriting \eqref{rholambda1} yields
\begin{equation}\label{eq:2-1}
\begin{aligned}
    \rho_{\hat{x}^a_{t+1}}|u^d_{DM} &= \frac{|r^2-1+2q\sin \psi -2p (r-\cos\psi)|}{2\sqrt{ (r-\cos\psi)^2 + \sin ^2 \psi}}\\
    &> \frac{r^2-1 + 2q\sin \psi -r (r-\cos \psi)}{2\sqrt{ (r-\cos\psi)^2 + \sin ^2 \psi}}\\
    &= \frac{r-\cos\psi}{2\cos \psi \sqrt{ (r-\cos\psi)^2 + \sin^2 \psi}}.
\end{aligned}
\end{equation}
Here last equality again used $p<\frac{r}{2}$ and $q = \frac{r}{2}\tan \psi$.

Furthermore, 
\begin{equation}
    \frac{r-\cos\psi}{2\cos \psi \sqrt{ (r-\cos\psi)^2 + \sin^2 \psi}} \geq \frac{1}{2}
\end{equation}
holds whenever $r>\sqrt{2}$. Note that equality holds only when $\psi = 0$. Finally, Due to the rigidity of the transformation $\Phi  (x,y)$, we can obtain $\rho_{x^a_{t+1}}|u^d_{DM}\geq \rho_{x^a_{t+1}}|u^d_{PP}$ directly from $\rho_{\hat{x}^a_{t+1}}|u^d_{DM} \geq \rho_{\hat{x}^a_{t+1}}|u^d_{PP}$, completing the proof.

\end{proof}
\begin{remark}
Unlike the PP strategy, goal of which is to capture the attacker, the DM strategy has different goal to maximize the defense margin. Therefore, it is inadequate to conduct stability analysis for the DM strategy as in Theorem~\ref{theorem:main1}.
\oprocend
\end{remark}

This theorem asserts that the DM strategy takes the safe zone into account and therefore it can be used to complement the PP strategy. AN empirical extension of Theorem~\ref{theorem:main2} is discussed in Section~\ref{sec:Simulation}.

Although the DM strategy is expected to be better than the PP strategy in terms of defense margin, the defense margin is not guaranteed to be non-decreasing. Together with the fact that the DM strategy does not explicitly steer to intercept the defender, solely relying on the DM strategy can lead to a shrinking defense margin without an interception. 

\subsection{Combination: Adjusted Defense Margin Strategy}\label{sec:TV}
We have discussed the limitations of using the PP and the DM strategy in Sections~\ref{sec:PP} and \ref{sec:met} respectively. In this subsection, we introduce a parameterized combination of the two to better suit our problem \eqref{condition1}. 

We introduce a weight parameter $\lambda_t$ and define the \textit{Adjusted Defense Margin  (ADM) strategy} as follows:
\begin{equation}\label{combination}
    u^{d}_t = c\lambda_t \frac{y-x^{d}_t}{\|y-x^{d}_t\|} + c (1-\lambda_t) \frac{l_{\mu}-x^{d}_t}{\|l_{\mu}-x^{d}_t\|}
\end{equation}
where $c$ is a normalizing constant that makes $\|u^d_t\| = 1$. Note that $\lambda_t \equiv 1$, and $\lambda_t \equiv 0$ restores the PP and the DM strategy.

The idea of the ADM strategy is to follow the DM strategy until it gets to a favorable position to apply the PP strategy. The problem is to construct adequate $\lambda_t$ that effectively balances DM and PP strategy in the evolving dynamics of the mission.

We first define the reliability of an observation $y_t$ which we can utilize without precise knowledge of $x^a_t$.

\begin{definition}
Let $\hat{w}_t$ be an estimate of the uncertainty $w_t$, expressed as
\begin{equation}
    \hat{w}_t \sim \mathcal{N} (0, \beta (\|y_t-x^d_t\|^2 I_2).
\end{equation}
The reliability of the observation $y_t$ is denoted as $P_t$ and is expressed as follows:
\begin{equation} \label{eq:Pt}
    \begin{aligned} 
        P_t = F (k, k) &+ F (-k,-k) \\
        &-F (-k,k)-F (k,-k).
    \end{aligned}
\end{equation}
Here $F (\cdot,\cdot): \mathbb{R}\times\mathbb{R} \rightarrow \mathbb{R}$ is the cumulative distribution function  (CDF) of a multivariate Gaussian distribution $\hat{w}_t$, and $k$ is the characterizing length of the reliability square. 
\end{definition}

In the above definition, we are creating a square of length $2k$, the center of which is the observation $y_t$. Then we obtain a probability of the target being inside the square by integrating multivariate Gaussian distribution within the square while the variance of the distribution is $\beta\|y_t-x^d_t\|^2 I_2$. Moreover, we are trying to make the control decision based on current observation $y_t$, which is known as the \textit{certainty equivalence approach}\cite{mania2019certainty}. 

\begin{remark}
One might argue that it would be more mathematically rigorous to use integration over multivariate Gaussian distribution within a fixed radius instead of a square. However, \eqref{eq:Pt} is known to be a good approximation\cite{tanash2021improved, lopez2011versatile} and the computation can be done in time complexity of $O (1)$, which is a substantial advantage in real-time missions.
\end{remark}

Utilizing the reliability of an observation, we propose the following parameterization:
\begin{equation}\label{lambda_t}
\begin{aligned}
    \lambda_t = P_t.
\end{aligned}
\end{equation}

Rewriting \eqref{combination}, we obtain
\begin{equation}\label{TV_control}
    u^{d}_t = cP_t \frac{y-x^{d}_t}{\|y-x^{d}_t\|} + c (1-P_t) \frac{l_{\mu}-x^{d}_t}{\|l_{\mu}-x^{d}_t\|}.
\end{equation}

Consider a case without observational noise, i.e. $\beta = 0$. This automatically makes $P_t \equiv 1$, reducing \eqref{TV_control} to \eqref{pp_control}. In other words, if we have accurate information about the attacker at all times, we follow the PP strategy. When the uncertainty is large, we tend more to the DM strategy. 

\section{Simulation and Results}\label{sec:Simulation}
In this section, we present a performance comparison of the strategies discussed in this paper. To this end, we first describe the parameters and different types of attacker behaviors used in the simulation. Then we provide and explain the simulation results. To obtain further insight, we introduce a function approximator of $\Delta \rho_{x^a_t}$ using Neural Network (NN) to extend the result of the Theorem~\ref{theorem:main2}. 

Simulations of the defenders using the PP strategy \eqref{pp_control}, the DM strategy \eqref{new_control}, and the ADM strategy \eqref{TV_control} were tested against three different behaviors of attackers: \textbf{1) Linear, 2) Spiral}, and \textbf{3) Intelligent}. Linear and spiral behaviors are predefined controls, and an intelligent attacker behaves in reaction to the defender. Linear strategy is a strategy that simply steers to the origin, or $u^a_{Linear} = -\frac{x^a}{\|x^a\|}$. Specific details and algorithms for the other two behaviors are explained in the Algorithm~\ref{alg:attacker}. Figure \ref{fig:trajectory} gives an intuition of how attacker behaviors are designed.

\begin{algorithm}
\caption{Attacker behaviors}\label{alg:attacker}
\begin{algorithmic}
\State \textbf{Spiral behavior}
    \State $x,y \gets x^a$
    \State $r \gets \|x^a\|$ 
    \State $\phi \gets \arctan y/x$
    \State $d\phi \gets 1/r$ 
    \State $\phi \gets \phi - d\phi$
    \State $x_2 \gets  (r-1)\cos d\phi - x$
    \State $y_2 \gets  (r-1)\sin d\phi - y$
    \State $u^a \gets [x_2-x, y_2-y]/\|[x_2-x, y_2-y]\|$ 
    \State $x^a \gets x^a+u^a$\\
\State \textbf{Intelligent behavior}
    \State $w \sim \mathcal{N} (0,\beta \|e_t\|^2 I_2)$
    \State $\hat{x}^d \gets x^d +w$ \Comment{Noisy observation} 
    \State $d_1 \gets - (\hat{x}^d-x^a)/\|\hat{x}^d-x^a\|$ \Comment{Evade defender}
    \State $d_2 \gets -x^a /\|x^a\|$ \Comment{Steer to the origin}
    \State $k_1 \gets 1/\|\hat{x}^d-x^a\|$
    \State $k_2 \gets 1$
    \State $x^a \gets x^a +  (k_1 d_1 + k_2 d_2)/\|k_1 d_1 + k_2 d_2\|$
\end{algorithmic}
\end{algorithm}

\begin{table}[ht]
\centering
\caption{Parameters used in Simulation}\label{table:parameters}
\begin{tabularx}{0.48\textwidth} { 
  >{\setlength\hsize{0.35\hsize}\centering\arraybackslash}X 
  | >{\setlength\hsize{1.5\hsize}\raggedright\arraybackslash}X
  | >{\setlength\hsize{0.8\hsize}\centering\arraybackslash}X}
\hline
\hline
Notation & Description & Value  \\ 
\hline
$t_f$ & Terminal time & $\infty$\\
$R_{\Omega_I}$ & Radius of the region of interest & 50  \\
$R_{\Omega_S}$ & Radius of the safe zone & 10 \\
$\tau$ & Maximum range of interception & 2\\
$\Lambda^d$ & Distribution of a defender's initial position $x^d_{t_i}$ in polar coordinate system & $  (\mathcal{U}[0,20], \mathcal{U}[-\pi,\pi])$\\
$\Lambda^a$ & Distribution of an attacker's initial position $x^a_{t_i}$ in polar coordinate system& $ (\mathcal{U}[45,50], \mathcal{U}[-\pi,\pi])$\\
$\beta$ & Coefficient for variance of uncertainty & 0.05 \\
$k$ & Characterizing length of a square in \eqref{eq:Pt} & 0.5 \\
\hline
\hline
\end{tabularx}
\end{table}

Table~\ref{table:parameters} summarizes the parameters used in the simulation. In the table, $\Lambda^a$ and $\Lambda^d$ were introduced to randomly initialize the agents, and $\mathcal{U}$ denotes a uniform distribution. All the values are normalized, and therefore they are dimensionless.

\begin{figure}[ht]
    \centering
    \includegraphics[width = \linewidth]{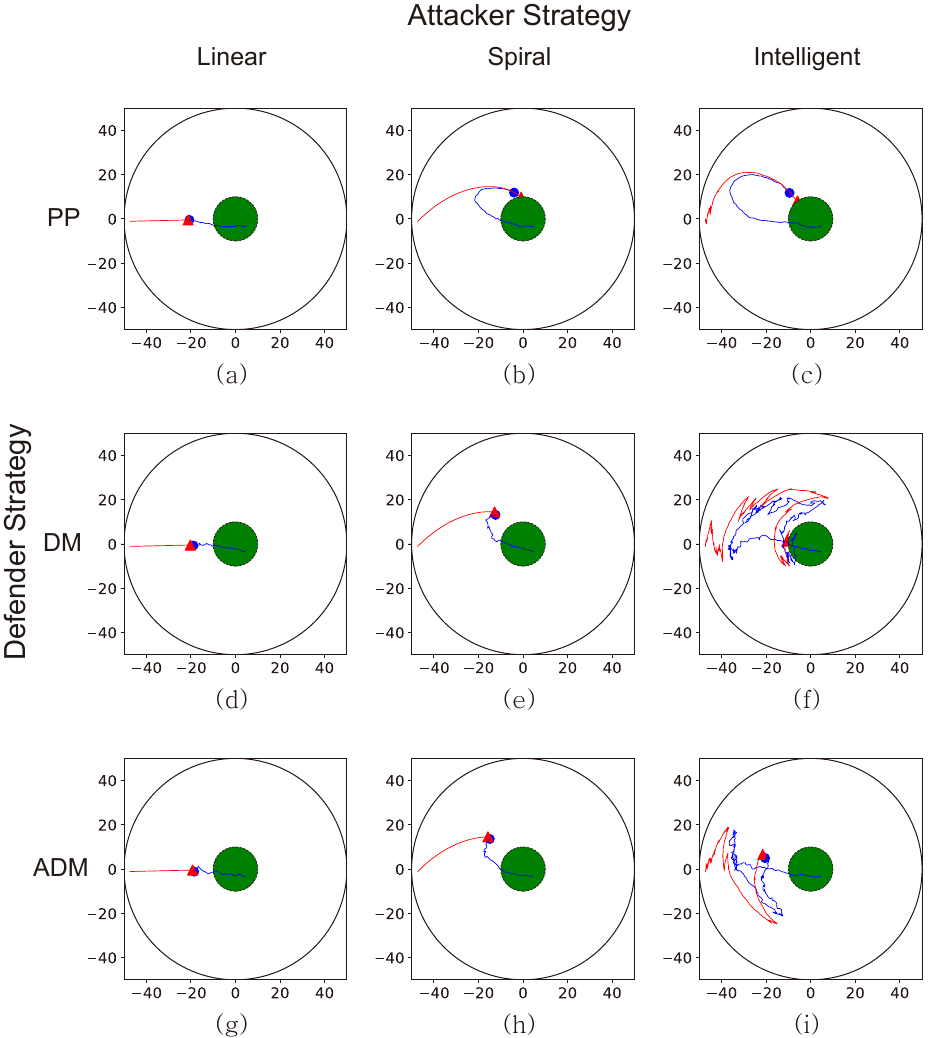}
    \caption{Trajectories of an attacker and a defender under various behavior scenarios. The outer circle represents $\Omega_I$ and the green circle represents $\Omega_S$. Red and blue lines respectively denote the attacker $x^a_t$ and the defender $x^d_t$ trajectories over time. All cases in this figure have homogeneous initial positions of an attacker and a defender $ (x^d_{t_i},x^a_{t_i})$. A defender has successfully defended the safe zone in all cases but (b), (c), and (f).}
    \label{fig:trajectory}
\end{figure}

Figure~\ref{fig:trajectory} illustrates a trajectory history of nine different case scenarios obtained by three defender strategies and three attacker strategies. The result shows that the PP guidance law fails to defend the safe zone against spiral (b) and intelligent (c) attacker behavior. The DM strategy failed against an intelligent attacker (f), and finally, our proposed ADM strategy has successfully defended all types of attackers. 

Figure~\ref{fig:winpercent_barchart} shows a defense performance obtained by running seeded 1,000 trials for each scenario. 
\begin{figure}[h]
    \centering
    \includegraphics[width =0.9\linewidth]{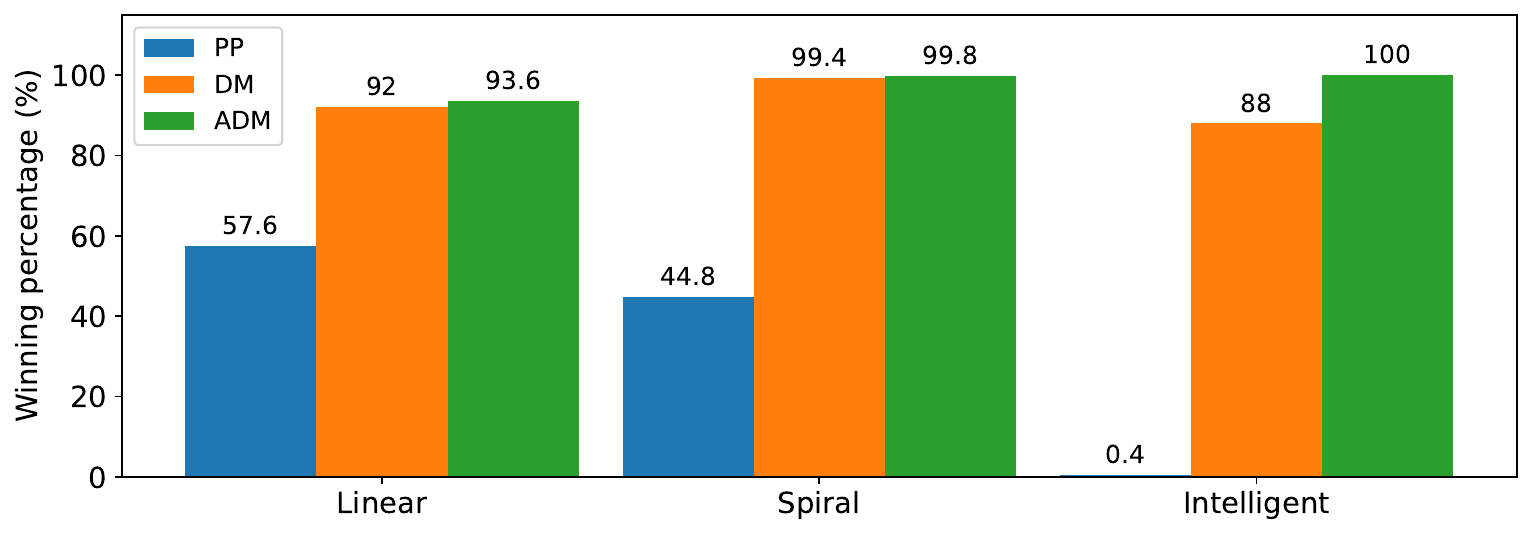}
    \caption{Winning percentage of each strategy against different attacker strategies. Each percentage result is obtained by counting successful missions out of 1,000 randomly initialized scenarios.}
    \label{fig:winpercent_barchart}
    \vspace{-0.2cm}
\end{figure}

In Figure~\ref{fig:winpercent_barchart}, we can see that the performance of the PP strategy drops as the attacker behavior complexity increases, whereas the DM strategy peaked its performance against the Spiral attacker behavior. Lastly, the ADM strategy performed better with higher complexity of attacker behavior.

The performance drop of the PP strategy can be explained with Theorem~\ref{theorem:main1}. For instance, Figures~\ref{fig:trajectory} (b) and (c) show that once $e_t^\top u^a_t \simeq 1$, the PP strategy becomes helpless. 

The DM strategy is designed to complement such limitations of the PP strategy. As a result, the DM achieves higher performance than the PP strategy in all scenarios. However, we can see its limitation in Figure~\ref{fig:trajectory} (f), in which its passivity culminated in mission failure. Due to such cases, the performance of the DM strategy drops against an intelligent attacker. 

The ADM strategy combines the two strategies and successfully improves performance. Some of the failure cases in a non-intelligent attacker can be explained by the randomness of the initial conditions. Some bipolar initial conditions of an attacker and a defender guaranteed the attacker to win the mission regardless of a defender's strategy. 

Next, we aim to extend Theorem~\ref{theorem:main2} via numerical simulation. Theorem~\ref{theorem:main2} assumed absence of uncertainty $w_t$ for simplicity. Here we use a Neural Network to approximate for $\Delta \rho_{x^{a}_t}$ and analyze results without such assumption.

\begin{definition}\label{def:robust1}
Assume $u^a_t = -\frac{x^a_t}{\|x^a_t\|}$. A strategy $u^d_A$ is said to be safer than strategy $u^d_B$ with respect to uncertainty if
\begin{equation}\label{safe_probability}
    \Delta (\rho_{x^{a}_t})|_{u^d_A} \geq \Delta (\rho_{x^{a}_t})|_{u^d_B}
\end{equation}
where $u^d_A$ and $u^d_B$ rely on uncertain observation of $x^a_t$.
\end{definition}

To compare $\Delta \rho_{x^a_t}|u^d_{PP}$ and $\Delta \rho_{x^a_t}|u^d_{DM}$ in the presence of uncertainty, we trained two fully-connected neural networks $\Delta\rho_{\hat{x}^a_t}|u^d_{PP}$ and $\Delta\rho_{\hat{x}^a_t}|u^d_{DM}$, both of which take $ (x^a_t,x^d_t)$ as an argument and respectively return a prediction of $\Delta \rho_{x^a_t}|u^d_{PP}$ and $\Delta \rho_{x^a_t}|u^d_{DM}$. The training data $\{(x^a_i,x^d_i),\Delta \rho_{x^a_i}\}_i^{N}$ is collected while simulating to obtain Figure~\ref{fig:winpercent_barchart} against a linear attacker behavior. Here $N$ is the total number of data collected. FCNN has 2 hidden layers of 100 nodes and a learning rate of 0.001, and we used an Adam optimizer for our work. To test the model, we randomly generated test data samples $ (x^a_i,x^d_i)_i^{M}$, where $M$ is the number of test data samples. Here $x^a_i$ and $x^d_i$ are uniformly sampled from a circle of random radius $\|x^a_t\| \sim \Lambda^a = \mathcal{U}[25,40]$ and $\|x^d_t\| \sim \Lambda^d =\mathcal{U}[0,15]$, respectively. In this work we used $M=100,000$.

The following table shows the result of the test:
\begin{center}
\begin{tabular}{ c|c|c } 

 & PP & DM \\
\hline
$\Delta \rho_{\hat{x}^a_t}$ & -0.132 & -0.028 \\ 

\end{tabular}
\end{center}

This again shows that $\Delta \rho_{x^a_t}|u^d_{DM}\geq \Delta \rho_{x^a_t}|u^d_{PP}$ even in the presence of uncertainty $w_t$.

\section{Conclusion}\label{Conclusion}
This work introduces a new metric called defense margin to solve the problem of a protective mission in which observation of the rogue attacker is noisy. We provided analytical proof to justify the implementation of the control strategy based on the defense margin. Finally, empirical results validate the efficacy of the strategy. 

Future research avenues shall include methods to tune optimal parameters according to the sensitivity of each parameter concerning the defense performance. Extension to a multi-agent problem, adding obstacles in the environment, and implementation in a 3D environment can also be part of future work. Lastly, various attacker behaviors could be designed and implemented for more comprehensive and reliable simulation results. 

\bibliographystyle{IEEEtran} 
\bibliography{refs} 
\end{document}